\documentclass[aps,pre,reprint,twocolumn,notitlepage,superscriptaddress,longbibliography]{revtex4-2}

\usepackage[dvipdfmx]{graphicx,color}	
\usepackage{amsmath}
\usepackage{amsfonts}
\usepackage{amssymb}	
\usepackage{amsthm}  
  
  \newtheorem*{thm}{Theorem}
  \newtheorem{prop}{Proposition}[section]

\usepackage{bm}
\usepackage{bbm}  
\usepackage{braket}
\usepackage{mathtools}
\usepackage{stmaryrd}
\usepackage{cases}
\usepackage{color}
\usepackage{physics}
\usepackage{array}  
\usepackage{float}  
\usepackage{framed}
\usepackage{tikz}  
\usetikzlibrary{intersections,calc,arrows.meta}  
\usepackage{comment}

\DeclareUnicodeCharacter{0300}{HEREHEREHERE}
\usepackage[driverfallback=dvipdfmx,setpagesize=false,
bookmarks=false,bookmarksopen=false,bookmarksnumbered=true,
pdfborder={0 0 0},linkcolor=blue,citecolor=blue,colorlinks=true]{hyperref} 

\graphicspath{{./figures/}}	

\newcommand{\diff}{\mathrm{d}}	
\newcommand{\Wass}{\mathcal{W}}
\newcommand{\Error}{\mathcal{E}}


\begin{document}
\color{black}
\title{Thermodynamic optimization of finite-time feedback protocols for Markov jump systems}
\author{Rihito Nagase}
\affiliation{Department of Applied Physics, The University of Tokyo, 7-3-1 Hongo, Bunkyo-ku, Tokyo 113-8656, Japan}
\author{Takahiro Sagawa}
\affiliation{Department of Applied Physics, The University of Tokyo, 7-3-1 Hongo, Bunkyo-ku, Tokyo 113-8656, Japan}
\affiliation{Quantum-Phase Electronics Center (QPEC), The University of Tokyo, 7-3-1 Hongo, Bunkyo-ku, Tokyo 113-8656, Japan}
\begin{abstract}
In recent advances in finite-time thermodynamics, optimization of entropy production required for finite-time information processing is an important issue. In this work, we consider finite-time feedback processes in classical discrete systems described by Markov jump processes, and derive achievable bounds on entropy production for feedback processes controlled by Maxwell's demons. The key ingredients of our approach is optimal transport theory and an achievable Fano's inequality, by which we optimize the Wasserstein distance over final distributions under fixed consumed information. Our study reveals the minimum entropy production for consuming a certain amount of information, and moreover, the optimal feedback protocol to achieve it. These results are expected to lead to design principles for information processing in various stochastic systems with discrete states.
\end{abstract}
\maketitle


\section{Introduction}\label{sec_1}
One of the key concepts in stochastic thermodynamics is entropy production~\cite{StochasticThermo,Seifert2012stochastic,landi2021irreversible}, which quantifies entropy production accompanying irreversibility of dynamics. Its fundamental bound is imposed by the second law of thermodynamics, which can be achieved in the quasi-static limit requiring infinite time. Beyond the quasi-static regime, recent advancements in stochastic thermodynamics have established various bounds on entropy production by incorporating finite-time effects, represented by thermodynamic speed limits~\cite{Aurell_JStatPhys_2012,Shiraishi_Funo_Saito_PhysRevLett.121.070601_2018,Ito_PhysRevLett.121.030605_2018,Ito_Dechant_PhysRevX.10.021056_2020,Vo_VanVu_Hasegawa_PhysRevE.102.062132_2020,Plata_PhysRevE.101.032129_2020,Falasco_Esposito_PhysRevLett.125.120604_2020,Yoshimura_Ito_PhysRevLett.127.160601_2021} and thermodynamic uncertainty relations~\cite{barato2015thermodynamic,gingrich2016dissipation,pietzonka2016universal,shiraishi2016universal,horowitz2017proof,dechant2018multidimensional,brandner2018thermodynamic,hasegawa2019uncertainty,koyuk2020thermodynamic,horowitz2020thermodynamic,liu2020thermodynamic}. In particular, the speed limits based on optimal transport theory~\cite{Jordan_Kinderlehrer_Otto_SIAM_1998,Benamou_Brenier_2000ACF,Maas_JFuncAnaly_2011gradient,Auruell_PhysRevLett.106.250601_2011,Auruell_JStatPhys_2012,Dechant2019thermodynamic,Chen_Georgiou_IEEE_2017matricial,Chen_Georgiou_IEEE_2019stochastic,Fu_Georgiou_Automatica_2021maximal,Miangolarra_IEEE_2022geometry,VanVu_Hasegawa_PhysRevLett.126.010601,Dechant_JPhys_2022,Chennakasavalu_PhysRevLett.130.107101_2023,T.V.Vu_PhysRevX.13.011013} provide achievable bounds for any finite operation times, by explicitly identifying the optimal protocols.\par
Recently, there has been progress in applying these finite-time frameworks to thermodynamics of information~\cite{Parrondo_Horowitz_Sagawa_NatPhys_2015}.
In the finite-time regime, the speed limits and the thermodynamic uncertainty relations have been extended to situations incorporating information processing~\cite{Zulkowski_PhysRevE.89.052140_2014,Proesmans_PhysRevE.102.032105_2020,Proesmans_PhysRevLett.125.100602_2020,Nakazato_Ito_PhysRevResearch.3.043093_2021,tanogami2023universal,xia2024efficiency,fujimoto2024game,nagase2024thermodynamically,Kamijima2024finite,Kamijima2024optimal}. In particular, optimal transport theory provides achievable bounds for the finite-time entropy production for various information processing including measurement, feedback and information erasure, when the initial and final probability distributions are fixed~\cite{Zulkowski_PhysRevE.89.052140_2014,Nakazato_Ito_PhysRevResearch.3.043093_2021,fujimoto2024game,Kamijima2024finite,Kamijima2024optimal}. However, optimizing the finite-time entropy productions for a given amount of processed information has been addressed only for  measurement processes~\cite{nagase2024thermodynamically}, and the optimal entropy productions for consuming a given amount of information through feedback processes have not been elucidated. 
\par

In this study, we consider optimization of feedback processes in classical discrete systems obeying Markov jump processes. Specifically, we determine the optimal entropy production for fixed consuming mutual information $|\Delta I|$. The optimization consists of two stages. The first stage is based on optimal transport theory, where we optimize entropy production over time-dependent protocols under fixed initial and identify final distributions~\cite{T.V.Vu_PhysRevX.13.011013}. In the second stage, we further optimize entropy production over final distributions under fixed consumed information $|\Delta I|$, and obtain the achievable lower bounds on entropy production and the optimal final distributions. These are derived by exploiting an achievable Fano's inequality~\cite{Sakai_e22030288_2020}, which bounds the conditional Shannon entropy by a function of error rate. Our results are relevant to experimental platforms such as single electron systems~\cite{koski2014experimental,koski2015chip,chida2017power}.\color{black}\par
The organization of this paper is as follows. In Sec.~\ref{sec_setup}, we describe our setup along with a brief review of thermodynamics of information, optimal transport theory and thermodynamical speed limits. In Sec.~\ref{sec_3}, we introduce an achievable Fano's inequality as a mathematical tool, and describe our main mathematical Theorem, and provide its proof. In Sec.~\ref{sec_4}, we present two bounds on entropy production of feedback process as physical consequences of the main Theorem. We also identify the optimal protocols to achieve these bounds. In
Sec.~\ref{sec_5}, we numerically demonstrate our bounds and the optimal protocols by a coupled two-level system. In Sec.~\ref{sec_6}, we summarize the
results of this paper and discuss future prospects.\color{black}
\section{Setup}\label{sec_setup}
\subsection{Thermodynamics of information}  
We consider a bipartite classical stochastic system consisting of subsystems \(X\) and \(Y\), which take discrete states. The entire system is in contact with a heat bath at inverse temperature \(\beta\), and its dynamics is described by a Markov jump process. System \(X\) is the target of feedback control, while \(Y\) plays the role of Maxwell's demon applying feedback operations to \(X\) based on the measurement results on \(X\)'s state.  

The set of possible states for \(X\) is \(\mathcal{X} \coloneqq \{1, 2, \dots, n\}\), and for \(Y\) is \(\mathcal{Y} \coloneqq \{1, 2, \dots, n\}\). Here \(y\in\mathcal{Y}\) corresponds to the measurement result on \(x\in\mathcal{X}\). The joint state of the entire system \(XY\) is represented by the pair \((x, y)\in\mathcal{X}\times\mathcal{Y}\), and the probability to find the system in state \((x, y)\) at time \(t\) is denoted as \(p_t^{XY}(x, y)\). The marginal probabilities for \(X\) and \(Y\) are defined as \(p_t^X(x) = \sum_y p_t^{XY}(x, y)\) and \(p_t^Y(y) = \sum_x p_t^{XY}(x, y)\), respectively.  

Since we focus on the feedback process, we assume that $X$ and $Y$ are correlated at the initial time $t=0$, as a consequence of the measurement performed beforehand. Here, we suppose that $Y$ is a controller that stores and uses the information of the target system $X$. Moreover, we make a simple assumption that the measurement is error-free. That is, the initial distribution is given of the form 
\begin{align}
  p_0^{XY} (x,y) = p^X(x) \delta_{x,y}\label{eq_1}
\end{align}
with $\delta_{x,y}$ being the Kronecker delta, which guarantees that $x=y$ holds with unit probability. During the feedback, the state $y$ of $Y$ is assumed to remain unchanged, i.e., no transitions occur between different states in $Y$. That is, transitions from \((x', y')\) to \((x, y)\) are prohibited if \(y \neq y'\). This assumption is reasonable for our classical stochastic processes, where the feedback operation on $X$ does not influence \(Y\)'s state itself \cite{Parrondo_Horowitz_Sagawa_NatPhys_2015}. Under this assumption, the marginal distribution \(p_t^Y(y)\) is fixed to a certain distribution \(p^Y(y)\) throughout the process, and the initial joint distribution is given by \(p_0^{XY}(x, y) = \delta_{x, y} p^Y(y)\).  

The time evolution of the entire system during the feedback process, from \(t = 0\) to \(t = \tau\), is described by the master equation:  
\begin{align}
  \frac{\diff}{\diff t} p_t^{XY}(x, y) = \sum_{x' : (x', x) \in \mathcal{N}_y} & \Big[ R_t^{X|y}(x, x') p_t^{XY}(x', y) \notag \\
  & - R_t^{X|y}(x', x) p_t^{XY}(x, y) \Big], \label{eq_2}
\end{align}
where \(R_t^{X|y}(x, x')\) is the transition rate from \((x', y)\) to \((x, y)\) at time \(t\). This rate describes the feedback operation that \(Y\) applies to \(x\) when the measurement result is \(y\). The set \(\mathcal{N}_y\) denotes the pairs of different \(X\)'s states \((x', x)\) that are allowed to transition when the measurement result is \(y\). These transitions are assumed to be bidirectional, meaning \((x', x) \in \mathcal{N}_y\) implies \((x, x') \in \mathcal{N}_y\). The stochastic heat \(Q_t^{X|y}(x, x')\) absorbed by \(X\) during a transition from \((x', y)\) to \((x, y)\) satisfies the local detailed balance condition:  
\begin{equation}
  \ln \frac{R_t^{X|y}(x, x')}{R_t^{X|y}(x', x)} = -\beta Q_t^{X|y}(x, x'). \label{eq_3}
\end{equation}

We next introduce mutual information, which represents the amount of information shared between \( X \) and \( Y \). The mutual information between \(X\) and \(Y\) at time \(t\) is defined as 
\begin{align}
  I_t^{X:Y} \coloneqq S(p_t^X) + S(p_t^Y) - S(p_t^{XY}), \label{eq_4}
\end{align}
where \(S(p)\) is the Shannon entropy of the probability distribution \(p\). If \(X\) and \(Y\) are uncorrelated, \(I_t^{X:Y} = 0\), while \(I_t^{X:Y} > 0\) otherwise. In our setup, with the assumptiuon (\ref{eq_1}), \(S(p_0^X) = S(p_0^Y) = S(p_0^{XY}) = S(p^Y)\), yielding \(I_0^{X:Y} = S(p^Y)\), which is the maximum mutual information for fixed \(p^Y\). This means that \( Y \) has fully acquired the information of \( X \) at time \( t = 0 \).

We also introduce entropy production, which represents  thermodynamic cost required for the feedback process. The entropy production from time \(0\) to $\tau$, denoted as \(\Sigma_\tau^{XY}\), is defined using the total heat absorbed by \(X\) up to time \(\tau\), \(Q_\tau^X \coloneqq \int_0^\tau \sum_{x,x',y} Q_t^{X|y}(x,x') R_t^{X|y}(x,x') p_t^{XY}(x',y) \, \mathrm{d}t\), as  
\begin{align}
    \Sigma_\tau^{XY} \coloneqq S(p_\tau^{XY}) - S(p_0^{XY}) - \beta Q_\tau^X. \label{eq_5}
\end{align}  
Here, the first two terms represent the entropy change of system \(XY\), and the third term represents the entropy change of the heat bath. Since \(Y\) does not undergo transitions, \(Q_\tau^X\) equals the total heat absorbed by the system \(XY\). Thus, \(\Sigma_\tau^{XY}\) represents the entropy change of the entire system including the heat bath, and describes dissipation due to irreversibility.\par
We here introduce the probability of a transition from state \((x',y)\) to \((x,y)\) denoted as \( j_t^{X|y}(x,x') \coloneqq R_t^{X|y}(x,x')p_t^{XY}(x',y) \), and the probability current from state \((x',y)\) to \((x,y)\) defined as \(J_t^{X|y}(x,x') = j_t^{X|y}(x,x') - j_t^{X|y}(x',x)\). We also define thermodynamic force for the transition from \((x',y)\) to \((x,y)\) as
\begin{equation}\label{eq_1_39}
  F_t^{X|y}(x,x')\coloneqq\ln\frac{j_t^{X|y}(x,x')}{j_t^{X|y}(x',x)}.
\end{equation}
Then, the entropy production rate $\sigma_t^{XY}\coloneqq\diff\Sigma_t^{XY}/\diff t$ can be expressed as
\begin{equation}\label{eq_1_40}
  \sigma_t^{XY}=\sum_{x>x'}J_t^{X|y}(x,x')F_t^{X|y}(x,x'),
\end{equation}
which satisfies the second law of thermodynamics \(\sigma_\tau^{XY} \geq 0\).
\par
entropy production for subsystem \(X\), \(\Sigma_\tau^X\), is defined as  
\begin{align}
  \Sigma_\tau^X \coloneqq S(p_\tau^X) - S(p_0^X) - \beta Q_\tau^X.  
\end{align}
Using this and the change in the mutual information \(\Delta I_\tau^{X:Y} \coloneqq I_\tau^{X:Y} - I_0^{X:Y}\), the entropy production can be decomposed as \(\Sigma_\tau^{XY} = \Sigma_\tau^X - \Delta I_\tau^{X:Y}\) \cite{Sagawa_Ueda_PhysRevLett.102.250602_2009}. In the present setup, since the maximum mutual information is stored at the initial time, we have \(\Delta I_\tau^{X:Y} \leq 0\). Therefore, the decomposition can be rewritten as  
\begin{align}
    \Sigma_\tau^{XY} = \Sigma_\tau^X + |\Delta I_\tau^{X:Y}|. \label{eq_6}
\end{align}  
By substituting this decomposition into the second law, we obtain 
\begin{align}
\Sigma_\tau^X\geq -\abs{\Delta I_\tau^{X:Y}}.\label{eq_6.5}
\end{align}
\color{black}
This indicates that by consuming mutual information, it is possible to achieve a lower entropy production than that determined by the second law for the case where \(Y\) is absent, \(\Sigma_\tau^X \geq 0\). We emphasize that the equality in (\ref{eq_6.5}) is achieved in the quasi-static limit, which requires infinite time \cite{Parrondo_Horowitz_Sagawa_NatPhys_2015}.

\subsection{Optimal transport theory and the speed limits}
We next briefly overview optimal transport theory. A key insight of optimal transport theory is that the minimized transport cost, known as the Wasserstein distance, serves as a metric between distributions. Optimal transport theory finds applications in fields such as image processing \cite{Haker_Zhu_Tannenbaum_Angenent_IJCV_2004}, machine learning \cite{Kolouri_7974883_2017}, and biology \cite{Schiebinger_Cell_2019}. Applying it to thermodynamics reveals that the achievable speed limits for entropy productions can be expressed in terms of the Wasserstein distance between initial and final distributions, enabling the identification of optimal protocols that achieve the equality.\par
The Wasserstein distance between two probability distributions \( p_0^{XY} \) and \( p_\tau^{XY} \) is defined as 
\begin{widetext}
\begin{align}
    \Wass(p_0^{XY}, p_\tau^{XY}) \coloneqq \min_{\pi \in \Pi(p_0^{XY}, p_\tau^{XY})} \sum_{x, x', y} d_{(x, y), (x', y)} \pi_{(x, y), (x', y)}, \label{eq_7}
\end{align}
\end{widetext}
where \(\pi_{(x, y), (x', y)} (\geq 0)\) represents the probability transported from state \((x', y)\) to state \((x, y)\). The collection of these probabilities for all states, \(\pi = \{\pi_{(x, y), (x', y)}\}_{x, x', y}\), is referred to as the transport plan. We denote by \(\Pi(p_0^{XY}, p_\tau^{XY})\) the set of all transport plans that transform the initial distribution \(p_0^{XY}\) into the final distribution \(p_\tau^{XY}\). This set satisfies the probability conservation conditions:  
\(\sum_{x \in \mathcal{X}} \pi_{(x, y), (x', y)} = p_0^{XY}(x', y)\),\ \(\sum_{x' \in \mathcal{X}} \pi_{(x, y), (x', y)} = p_\tau^{XY}(x, y)\). The coefficient \(d_{(x, y), (x', y)}\) is the minimum number of transitions required for moving from state \((x', y)\) to state \((x, y)\) under the Markov jump process defined in Eq.~(\ref{eq_2}), which determines the cost of transport per unit probability. Thus, the definition in Eq.~(\ref{eq_7}) indicates that by fixing the initial distribution \(p_0\) and the final distribution \(p_\tau\), the minimum total cost over all possible transport plans \(\pi\) can be defined as the distance between the distributions. The transport plan \(\pi\) that achieves this minimum cost is referred to as the optimal transport plan.\par
When probability distribution \( p_0^{XY} \) evolves into \( p_\tau^{XY} \) following Eq.~(\ref{eq_2}) from time \( t = 0 \) to \( t = \tau \), the entropy production is bounded as~\cite{T.V.Vu_PhysRevX.13.011013}  
\begin{equation}
    \Sigma_\tau^{XY} \geq \Wass\left(p_0^{XY}, p_\tau^{XY}\right) f\left(\frac{\Wass\left(p_0^{XY}, p_\tau^{XY}\right)}{D \tau}\right), \label{eq_8}
\end{equation}  
where \( f(x) \) is a function determined by the choice of fixed timescale \( D \). \par
We introduce two quantities representing the timescale: activity and mobility. We define activity \(a_t\) as
\begin{equation}\label{eq_1_41}
  a_t\coloneqq\sum_{y,\ x\neq x'}j_t^{X|y}(x,x')=\sum_{x>x'}a_t^{X|y}(x,x').
\end{equation}
Here, the local activity between states \((x, y)\) and \((x', y)\) is defined as \(a_t^{X|y}(x, x') \coloneqq j_t^{X|y}(x, x') + j_t^{X|y}(x', x)\), which corresponds to the number of jumps per unit time at time \(t\). Additionally, the time integral representing the total number of jumps from time \(0\) to \(\tau\) is expressed as \(A_\tau \coloneqq \int_0^\tau a_t \, \mathrm{d}t\), and its time average is given by \(\langle a \rangle_\tau \coloneqq A_\tau / \tau\). We define mobility $m_t$ as
\begin{equation}\label{eq_1_42}
  m_t\coloneqq\sum_{x>x'}\frac{J_t^{X|y}(x,x')}{F_t^{X|y}(x,x')}.
\end{equation}
This represents the probability currents induced by the thermodynamic forces applied to the system. As well as activity, the time integral is expressed as \(M_\tau \coloneqq \int_0^\tau m_t \, \mathrm{d}t\), and the time average is given by \(\langle m \rangle_\tau \coloneqq M_\tau / \tau\). \par 
 Depending on whether \( D \) corresponds to either of these two quantities, the function \( f \) in the lower bound (\ref{eq_8}) changes. If \( D \) represents the time-averaged mobility \(\langle m \rangle_\tau\), then \( f(x) = x \); if \( D \) represents the time-averaged activity \(\langle a \rangle_\tau\), then \( f(x) = 2 \tanh^{-1}(x) \).  

For either choice of \( D \), the optimal protocol \(\{R_t^{X|y}(x,x')\}\) that achieves the equality in inequality (\ref{eq_8}) can be constructed by the condition that the probability is transported from \( p_0^{XY} \) to \( p_\tau^{XY} \) along the optimal transport plan with a uniform and constant thermodynamic force.\par

By substituting the decomposition (\ref{eq_6}) into inequality (\ref{eq_8}), we obtain the fundamental bound on energy for the fixed initial probability distribution \( p_0^{XY} \) and the final distribution \( p_\tau^{XY} \)
\begin{align}
    \Sigma_\tau^X \geq -\abs{\Delta I_\tau^{X:Y}} + \Wass\left(p_0^{XY}, p_\tau^{XY}\right) f\left(\frac{\Wass\left(p_0^{XY}, p_\tau^{XY}\right)}{D \tau}\right). \label{eq_8.5}
\end{align}  
The equality is achieved by the same protocol that achieves inequality (\ref{eq_6}).\par
The minimum entropy production for consuming a certain amount of information via feedback processes is not determined by inequality (\ref{eq_8.5}) alone. This is because, for a fixed consumed information $|\Delta I_\tau^{X:Y}|$, there exist infinitely many possible final distributions $p_\tau^{XY}$, which have different Wasserstein distance \(\Wass\left(p_0^{XY}, p_\tau^{XY}\right)\). Therefore, we can further minimize the right-hand side of (\ref{eq_8.5}) which depends on \(\Wass\left(p_0^{XY}, p_\tau^{XY}\right)\) over $p_\tau^{XY}$ under fixed \( |\Delta I_\tau^{X:Y}| \), and determine the truly minimum entropy production for consuming $|\Delta I_\tau^{X:Y}|$.
%
%
%
%
%
%
\section{Fundamental bounds on the consumed information}\label{sec_3}
\subsection{Fano's inequality}
To address the above problem, we first introduce the mathematical tool called Fano's inequality. It enables us to evaluate the error that arises when a sender \( Y \) sends a message to a receiver \( X \) through a classical communication channel. Let the set of messages to be transmitted and received be \(\{1, 2, \dots, n\}\). Due to the presence of errors in the communication channel, when the sender \( Y \) sends a message \( y \), the receiver \( X \) may not always receive \( y \) with probability 1. The performance of the communication channel is characterized by the conditional probability distribution \( p^{Y|x}(y) \coloneqq p^{XY}(x, y)/p^X(x) \), which indicates the ambiguity about the original message \( y \) given the received message \( x \).

Consider the case where the sender \( Y \) stochastically transmits messages according to a distribution \( p^Y(y) \). Denote the joint probability that \( Y \) sends \( y \) and \( X \) receives \( x \) as \( p^{XY}(x, y) \). Now we aim to evaluate the error probability \( \Error \coloneqq \mathrm{Prob}(x \neq y)=\sum_{x \neq y} p^{XY}(x, y) \)) based on the channel performance \( p^{Y|x}(y) \). This evaluation is provided by Fano's inequality, and the traditional form is given by
\begin{equation}\label{eq_9}
    S^{Y|X} \leq -\Error \ln \Error - (1-\Error) \ln (1-\Error) + \Error \ln (n-1).
\end{equation}  
Here, the left-hand side represents the conditional Shannon entropy \( S^{Y|X} \coloneqq \sum_x p^X(x) S(p^{Y|x}) \), which quantifies the randomness of the communication channel's performance \( p^{Y|x}(y) \). The right-hand side is an increasing function of the error probability \(\Error\). Thus, Fano's inequality indicates that the randomness \( S^{Y|X} \) of the channel's performance leads to a higher error probability. The equality holds when the conditional probabilities \( p^{Y|x}(y) \) satisfies  
\begin{align}\label{eq_10}
    p^{Y|x}(y) =  
    \begin{cases} 
        1-\Error, & x = y, \\ 
        \Error/(n-1), & x \neq y.
    \end{cases}
\end{align}  However, depending on the probability distribution \( p^Y \), it might not be possible to construct \( p^{Y|x}(y) \) in accordance with Eq.~(\ref{eq_10}). Therefore, the traditional Fano inequality (\ref{eq_9}) is not generally achievable. 

On this matter, a previous research \cite{Sakai_e22030288_2020} has introduced a tighter version of Fano's inequality that can be achieved for any \( p^Y \). Applying it to the above setup of communication channel, we can summarize the achievable Fano's inequality as follows. For fixed \( p^Y \) and \( \Error \), there exists a probability distribution \( \tilde{p}_\Error^Y \) (whose construction will be explained below) determined by \( p^Y \) and \( \Error \), such that 

\begin{equation}\label{eq_11}
  S^{Y|X} \leq S\left(\tilde{p}_\Error^Y\right),
\end{equation}
which is achievable for any \( p^Y \).

Here, the construction of \( \tilde{p}_\Error^Y \) is as follows:
\begin{figure*}[tbp]
\centering
\includegraphics[scale=0.8]{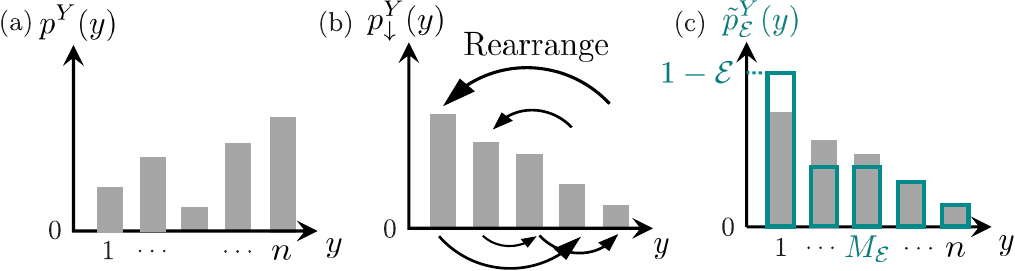}
\caption{The construction of \( \tilde{p}_\Error^Y \). (a) Histogram of \( p^Y \), the original probability distribution of \( Y \). (b) Histogram of \( p_\downarrow^Y \), the probability distribution obtained by rearranging \( p^Y \) in descending order. (c) Histogram of \( \tilde{p}_\Error^Y \).}
\label{fig_1}
\end{figure*}
\begin{enumerate}
  \item[(i)] Let $p_\downarrow^Y$ be the probability distribution obtained by rearranging $p^Y$(Fig.~\ref{fig_1}(a)) in descending order $p_\downarrow^Y(1)\geq p_\downarrow^Y(1)\geq\cdots\geq p_\downarrow^Y(n)$(Fig.~\ref{fig_1}(b). If \( p_\downarrow^Y(1) \geq 1 - \Error \), then \( \tilde{p}_\Error^Y \) is defined as \( \tilde{p}_\Error^Y = p_\downarrow^Y \). If \( p_\downarrow^Y(1) < 1 - \Error \), then \( \tilde{p}_\Error^Y \) is constructed according to the following steps.
  \item[(ii)] Define a function \( M_\Error \) that specifies a state of \( Y \) corresponding to \( \Error \) (the construction of the function \( M_\Error \) will be provided in step (iii)). Then, \( \tilde{p}_\Error^Y \) is defined as
  \begin{equation}
  \tilde{p}_\Error^Y(y)\coloneqq
  \begin{cases}
    1-\Error, & y=1\\
    \frac{\Error-\sum_{y=M_\Error+1}^np_\downarrow^Y(y)}{M_\Error-1}, & 2\leq y\leq M_\Error \\
    p_\downarrow^Y(y), & M_\Error<y\leq n.\label{eq_12}
  \end{cases}
  \end{equation}
  The corresponding histogram is shown in Fig.~\ref{fig_1}(c).
  \item[(iii)] \( M_\Error \) is defined as the largest \( m \in \{1, \cdots, n\} \) that satisfies  
  \begin{align}
    \frac{\Error - \sum_{y=m+1}^n p_\downarrow^Y(y)}{m - 1} < p_\downarrow^Y(m). \label{eq_15}
  \end{align}
When \( p_\downarrow^Y(1) < 1 - \Error \), such \( M_\Error \) is uniquely determined in the range \( 2 \leq M_\Error \leq n \). This can be interpreted as the largest \( M_\Error \) that ensures the probability distribution \( \tilde{p}_\Error^Y \) [defined in Eq.~(\ref{eq_12})] is in descending order.
  \end{enumerate}
The achievable Fano's inequality (\ref{eq_11}) can be regarded as an inequality that evaluates the upper bound of the conditional Shannon entropy \( S^{X|Y} \) when the error probability \( \Error \) and the marginal distribution \( p^Y \) are fixed. Therefore, this inequality can be applied not only to the communication channel setting described above, but also to general joint systems \( XY \) characterized by any joint probability distribution \( p^{XY}(x, y) \).
\subsection{Main theorem}
To obtain the fundamental bounds on entropy production required to consume information by the feedback processes, we consider the minimization problem of the second term on the right-hand side of inequality (\ref{eq_8.5}) for a fixed \( |\Delta I_\tau^{X:Y}| \). Since this term is a monotonically increasing function of \( \Wass \) for both $f(x)=x$ and $f(x)=2\tanh^{-1}(x)$, it suffices to minimize \( \Wass\left(p_0^{XY}, p_\tau^{XY}\right) \).

Our main theorem is now stated as follows, which will be used to solve the above problem in the next section.
\begin{thm}
  For a fixed \( p^Y \) and \( \Wass = \Wass\left(p_0^{XY}, p_\tau^{XY}\right) \),
  \begin{align}\label{eq_16}
    |\Delta I_\tau^{X:Y}|\leq S\left(\tilde{p}_\Wass^Y\right)
  \end{align}
  holds regardless of the structure of the entire system $XY$. $\tilde{p}_\Wass^Y$ is determined by identifying $\Error=\Wass$ in Eq.~(\ref{eq_12}). The equality can be achieved by setting the final distribution to a certain distribution (the explicit form will be described in the following section) if the entire system \( XY \) satisfies the following condition (C):
  \begin{itemize}
  \item[(C)] For all states \( y \in Y \), the direct transition in system \( X \) from the state \( x = y \) to any other state \( x' (\neq y) \) is possible. In other words,the direct transition from state \( (y, y) \) to \( (x', y) \) is possible.
\end{itemize}
\end{thm}
In our setup, we assumed that at \(t=0\), \(x = y\) holds with probability 1 through an error-free measurement. Therefore, the condition (C) indicates that regardless of the measurement result \(y\), the direct transition from the initial state \(x = y\) to any other state \(x' (\neq y)\) is allowed in the feedback control. Inequality (\ref{eq_16}) provides an upper bound on the mutual information that can be consumed in a finite-time feedback process when the Wasserstein distance is fixed. Since the right-hand side, \( S\left(\tilde{p}_\Wass^Y\right) \), is an increasing function of \( \Wass \), this implies that the larger the Wasserstein distance between the initial and final distributions, the more mutual information can be consumed by the feedback prosess. \par
This Theorem comprises two statements: first, that inequality (\ref{eq_16}) holds for any system $XY$ which satisfies our setup; second, that if the system $XY$ meets the condition (C), the equality in (\ref{eq_16}) can be achieved by setting $p_\tau^{XY}$ to an optimal distribution. We give the proof of these two statements in the two following subsections respectively.
\subsection{Proof of inequality (\ref{eq_16})}
We here prove inequality (\ref{eq_16}). By transforming Eq.~(\ref{eq_4}), the mutual information can be rewritten as 
\[
I_t^{X:Y} = S\left(p_t^Y\right) - S_t^{Y|X},
\]
where \( S_t^{Y|X} \) denotes the conditional Shannon entropy of \( Y \) given \( X \) at time \( t \).  
In the present setting, \( p_t^Y = p^Y \) is fixed, and initially, \( x = y \) with probability 1. Therefore, \( S_0^{Y|X} = 0 \). Thus, we obtain
\begin{align}
  \abs{\Delta I_\tau^{X:Y}}&=S_\tau^{Y|X}.
\end{align}
Let the probability that \( x \neq y \) holds at time \( \tau \) be \( \Error_\tau \coloneqq \sum_{x \neq y} p_\tau^{XY}(x, y) \). From achievable Fano's inequality (\ref{eq_11}), we obtain
\begin{align}
  |\Delta I_\tau^{X:Y}| \leq S\left(\tilde{p}_{\Error_\tau}^Y\right). \label{eq_18}
\end{align}
Here, \( S\left(\tilde{p}_{\Error_\tau}^Y\right) \) is an increasing function of \( \Error_\tau \), and \( \Error_\tau \leq \Wass \) holds (the proofs are provided in the Appendix). Substituting these into inequality (\ref{eq_18}) yields the main Theorem (\ref{eq_16}).  

\subsection{Proof of the optimality}
We prove that the equality in (\ref{eq_16}) can be achieved by setting $p_\tau^{XY}$ to an optimal distribution if the entire system $XY$ satisfies the condition (C). First, we will describe the method for constructing the optimal final distribution \(p_\tau^{XY}\). Next, we will prove that this final distribution satisfies the constraints \(p_\tau^Y = p^Y\) and \(\Wass(p_0^{XY},p_\tau^{XY}) = \Wass\). Finally, we will show that this final distribution achieves the equality in (\ref{eq_16}).\par
The optimal final distribution \( p_\tau^{XY} \) can be constructed as follows. Since the joint probability distribution of the entire system can be expressed as \( p_\tau^{XY}(x, y) = p_\tau^X(x) p_\tau^{Y|x}(y) \), it suffices to define \( p_\tau^X(x) \) and \( p_\tau^{Y|x}(y) \). Here, let \( \sigma_{p^Y} \) denote the permutation that rearranges the states \( y \) in descending order of the probability distribution \( p^Y \). In other words, \( \sigma_{p^Y}(y) = n \) means that \( y \) corresponds to the \( n \)-th largest value of \( p^Y(y) \). If there are states with equal probabilities, their ordering in \( \sigma_{p^Y} \) can be chosen arbitrarily; once \( \sigma_{p^Y} \) is defined, it must be fixed throughout the following discussion.\par
The optimal \( p_\tau^X(x) \) is then defined as  
  \begin{align}\label{eq_19}
  p_\tau^X(x)=
  \begin{cases}
    \displaystyle\frac{p^Y(x)-\tilde{p}_\Wass^Y(2)}{\tilde{p}_\Wass^Y(1)-\tilde{p}_\Wass^Y(2)},\quad& 1\leq\sigma_{p^Y}(x)\leq M_\Wass,\\
    0,&\sigma_{p^Y}(x)> M_\Wass.
  \end{cases}
\end{align}
From the definition of \( \tilde{p}_\Wass^Y \), we have \( \tilde{p}_\Wass^Y(1) = 1 - \Wass \) and  
$\tilde{p}_\Wass^Y(2) = \frac{\Wass - \sum_{y=M_\Wass+1}^n p_\downarrow^Y(y)}{M_\Wass - 1}$. The optimal conditional probability distribution \( p_\tau^{Y|x}(y) \) is given as  
\begin{align}\label{eq_20}
  p_\tau^{Y|x}(y)=
  \begin{cases}
    \tilde{p}_\Wass^Y(1), & y=x,\\
    \tilde{p}_\Wass^Y(x),&\sigma_{p^Y}(y)=1,\\
    \tilde{p}_\Wass^Y(\sigma_{p^Y}(y)),&\mathrm{otherwise.}
  \end{cases}
\end{align}
\begin{figure*}[tbp]
\centering
\includegraphics[scale=0.9]{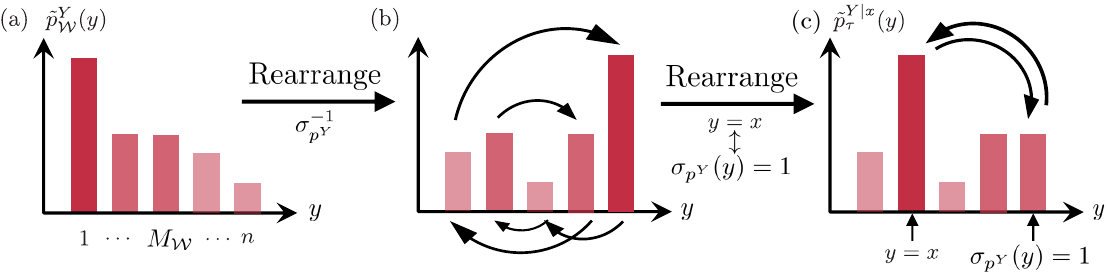}
\caption{(a) Histogram of \( \tilde{p}_\Wass^Y \). (b) Histogram of the probability distribution obtained by rearranging \( \tilde{p}_\Wass^Y \) such that the descending order of probability matches that in \( p^Y \). (c) Histogram of the optimal conditional probability distribution \( p_\tau^{Y|x} \), constructed by swapping the probabilities of the state with the highest probability in \( \tilde{p}_\Wass^Y \) and the state \( y = x \).}
\label{fig_2}
\end{figure*}
This procedure for constructing the optimal probability distribution can also be visualized graphically as follows: begin with  \( \tilde{p}_\Wass^Y \) (Fig.~\ref{fig_2} (a)), which by definition is a probability distribution arranged in descending order along the states \( y \). Apply the permutation \( \sigma_{p^Y}^{-1} \) to \( \tilde{p}_\Wass^Y \) (Fig.~\ref{fig_2} (b)), which rearranges the states such that the descending order of probabilities matches that of \( p^Y \). Then, swap the probability of the state \( y = x \) with the probability of the state with the largest probability (i.e., the state satisfying \( \sigma_{p^Y}(y) = 1 \)) (Fig.~\ref{fig_2} (c)). From the second operation, the optimal conditional probability distribution \( p_\tau^{Y|x} \) satisfies \( p_\tau^{Y|x}(x) = \tilde{p}_\Wass^Y(1) = 1 - \Wass \).\par
We next show that its marginal distribution matches the fixed \( p^Y \). For $y$ such that $\sigma_{p^Y}(y)>M_\Wass$ holds, from Eq.~(\ref{eq_20}), $p_\tau^{Y|x}(y)=\tilde{p}_\Wass^Y(\sigma_{p^Y}(y))=p^Y(y)$ holds independently of $x$, which yields $\sum_xp_\tau^{XY}(x,y)=p^Y(y)$. For $y$ such that $\sigma_{p^Y}(y)\leq M_\Wass$ holds, from Eq.~(\ref{eq_20}) and the definition of $\tilde{p}_\Wass^Y$, $p_\tau^{Y|x}(y)$ takes the value $\tilde{p}_\Wass^Y(1)$ for $y=x$ and $\tilde{p}_\Wass^Y(2)$ for $y\neq x$. Therefore,
\begin{align}
  \sum_{x=1}^np_\tau^{XY}(x,y)&=\sum_{x=1}^np_\tau^X(x)p_\tau^{Y|x}(y)\label{eq_21}\\
  &=p_\tau^X(y)\tilde{p}_\Wass^Y(1)+[1-p_\tau^X(y)]\tilde{p}_\Wass^Y(2)\label{eq_22}\\
  &=p^Y(y).\label{eq_23}
\end{align}
Here, we used Eq.~(\ref{eq_19}) for the transformation from Eq.~(\ref{eq_22}) to Eq.~(\ref{eq_23}).

We next show that under the condition (C), the Wasserstein distance between the constructed \( p_\tau^{XY} \) and \( p_0^{XY} \) equals the fixed value \( \Wass \). When the condition (C) is satisfied, \( \Error_\tau = \Wass\left(p_0^{XY}, p_\tau^{XY}\right) \) holds (the proof is provided in the Appendix). Therefore, we obtain
\begin{align}
  \Wass\left(p_0^{XY}.p_\tau^{XY}\right)&=\Error_\tau\\
  &=1-\sum_{x,y:x=y}p_\tau^{XY}(x,y)\\
  &=1-\sum_{x}p_\tau^X(x)p_\tau^{Y|x}(x)\\
  &=1-\tilde{p}_\Wass^Y(1)\\
  &=\Wass.
\end{align}
We finally show that the constructed \( p_\tau^{XY} \) is optimal, i.e., it achieves the equality in (\ref{eq_16}). From Eq.~(\ref{eq_20}), since \( S\left(p_\tau^{Y|x}\right) = S\left(\tilde{p}_\Wass^Y\right) \) holds regardless of \( x \),  
\begin{align}
  \abs{\Delta I_\tau^{X:Y}}&=S_\tau^{Y|X}\\
  &=\sum_xp_\tau^X(x)S\left(p_\tau^{Y|x}\right)\\
  &=S\left(\tilde{p}_\Wass^Y\right)
\end{align}
holds. Therefore, the equality in (\ref{eq_16}) can be achieved.\color{black}
\section{Fundamental bound on entropy production}\label{sec_4}
\subsection{Speed limit for fixed mobility}
In the previous section, we considered the maximization problem for the consumed information \( |\Delta I_\tau^{X:Y}| \) for fixed \( \Wass\left(p_0^{XY}, p_\tau^{XY}\right) \), whose solution is given by the main Theorem (\ref{eq_16}). Since the original problem of minimizing \( \Wass\left(p_0^{XY}, p_\tau^{XY}\right) \) for fixed \( |\Delta I_\tau^{X:Y}| \) corresponds to the dual problem, its solution is directly derived by the main Theorem. From this dual solution, we can obtain the fundamental bound on entropy production required to consume \( |\Delta I_\tau^{X:Y}| \).\par
First, the right-hand side of inequality (\ref{eq_16}), \( S\left(\tilde{p}_\Wass^Y\right) \), is a strictly increasing function of \( \Wass \) in the range \( 0 \leq \Wass \leq 1 - p_\downarrow^Y(1) \), and takes a constant value \( S\left(p^Y\right) \) for \( \Wass \geq 1 - p_\downarrow^Y(1) \) (the proof is provided in the Appendix). Therefore, given a fixed \( p^Y \), the inverse function \( \Wass_{p^Y}^{\min} : [0, S\left(p^Y\right)] \to [0, 1 - p_\downarrow^Y(1)] \) can be defined as \( S\left(\tilde{p}_{\Wass_{p^Y}^{\min}(I)}\right) = I \). This function provides the minimum Wasserstein distance \( \Wass\left(p_0^{XY}, p_\tau^{XY}\right) \) for the fixed amount of consumed information , \( |\Delta I_\tau^{X:Y}| \).\par
Here, considering Eq.~(\ref{eq_8.5}) with \( D \) chosen as the time-averaged mobility \( \langle m \rangle_\tau \), the right-hand side becomes an increasing function of \( \Wass\left(p_0^{XY}, p_\tau^{XY}\right) \). Therefore, minimizing \( \Wass\left(p_0^{XY}, p_\tau^{XY}\right) \) is equivalent to minimizing entropy production $\Sigma_\tau$. This allows us to derive the speed limit in the feedback process as
\begin{align}\label{eq_32}
  \Sigma_\tau^X\geq-\abs{\Delta I_\tau^{X:Y}}+\frac{{\Wass_{p^Y}^{\min}\left(\abs{\Delta I_\tau^{X:Y}}\right)}^2}{\tau\langle m\rangle_\tau},
\end{align}
which provides the minimum entropy production required to consume information \( |\Delta I_\tau^{X:Y}| \) through feedback processes keeping time-averaged mobility \( \langle m \rangle_\tau \) fixed.\par
The equality in inequality (\ref{eq_32}) is achieved by simultaneously satisfying the equalities in inequalities (\ref{eq_8}) and (\ref{eq_16}). First, the equality in (\ref{eq_16}) is  achieved by setting the final distribution \( p_\tau^{XY} \) to the form determined by Eqs.~(\ref{eq_19}) and (\ref{eq_20}), while satisfying condition (C). Next, the equality in (\ref{eq_8}) is achieved by constructing a protocol \(\{R_t^{X|y}(x,x')\}_{t=0}^\tau\) that transports the probability between \(p_0^{XY}\) and \(p_\tau^{XY}\) along the optimal transport plan, with uniform and constant thermodynamic force. Consequently, inequality (\ref{eq_32}) is satisfied. This highlights the thermodynamic implication of our main Theorem.
\subsection{Speed limit for fixed activity}
When we choose time-averaged activity \(\langle a \rangle_\tau\) as \(D\), the second term on the right-hand side of inequality (\ref{eq_8}) becomes \(2\Wass\tanh^{-1}[\Wass/(\tau\langle a \rangle_\tau)]\). Similarly to the case where mobility is fixed, this term is an increasing function of \(\Wass\) for any fixed \(\tau\) and \(\langle a \rangle_\tau\). Therefore, by applying the main Theorem, we can derive another speed limit
\begin{widetext}
  \begin{align}\label{eq_33}
  \Sigma_\tau^X\geq-\abs{\Delta I_\tau^{X:Y}}+2\Wass_{p^Y}^{\min}\left(\abs{\Delta I_\tau^{X:Y}}\right)\tanh^{-1}\frac{\Wass_{p^Y}^{\min}\left(\abs{\Delta I_\tau^{X:Y}}\right)}{\tau\langle a\rangle_\tau},
\end{align}
\end{widetext}
which provides the minimum entropy production required to consume the information \(\abs{\Delta I_\tau^{X:Y}}\) in feedback control when time-averaged activity \(\langle a \rangle_\tau\) is fixed. 

The equality in inequality (\ref{eq_33}) is achieved by simultaneously achieving the equalities in inequalities (\ref{eq_8}) and (\ref{eq_16}). First, the equality in (\ref{eq_16}) is realized by setting the final distribution \(p_\tau^{XY}\) to the form determined by Eqs.~(\ref{eq_19}) and (\ref{eq_20}), while satisfying condition (C). Next, the equality in (\ref{eq_8}) is achieved by constructing a protocol \(\{R_t^{X|y}(x,x')\}_{t=0}^\tau\) that transports the probability between the specified initial and final distributions \(p_0^{XY}\) and \(p_\tau^{XY}\) along the optimal transport plan, with uniform and constant thermodynamic force. Consequently, inequality (\ref{eq_33}) is achieved. This again highlights the thermodynamic implication of the main Theorem.
\section{Example: Coupled two level system}\label{sec_5}
We demonstrate our main inequalities (\ref{eq_32}) and (\ref{eq_33}) by numerical simulation of a coupled system which can take two states, \(x = 1, 2\) and \(y = 1, 2\). Such a setup may be realized, for instance, by confining colloidal particles to two positions corresponding to \(1\) and \(2\) using an optical potential, with coupling mediated by Coulomb interactions. In this case, if the time evolution of the entire system is governed by Eq.~(\ref{eq_2}), the controllable transition rates are given by the set \(\{R_t^{X|1}(2,1), R_t^{X|1}(1,2), R_t^{X|2}(2,1), R_t^{X|2}(1,2)\}\). Throughout this section, we represent the probability distribution as a matrix \(p_t^{XY} = \left[p_t^{XY}(x,y)\right]_{x,y}\).\par
First, when fixing \(p^Y(1) = p\ (< 1/2)\), the initial distribution \(p_0^{XY}\) is given as
\begin{align}
  p_0^{XY}&=\left[
\begin{array}{cc}
p & 0 \\
0 & 1-p
\end{array}\right].\label{eq_34}
\end{align}
The equality in the bound (\ref{eq_32}) is achieved when the equalities in both (\ref{eq_8}) and (\ref{eq_16}) are simultaneously satisfied. By setting $\Wass = \Wass_{p^Y}^\mathrm{min}\left(\abs{\Delta I_\tau^{X:Y}}\right)$ with respect to the consumed mutual information $|\Delta I_\tau^{X:Y}|$, the optimal final distribution that satisfies the equality in (\ref{eq_16}) is obtained as
\begin{align}
p_\tau^{XY}&=\left[
\begin{array}{cc}
\frac{p-\Wass}{1-2\Wass}(1-\Wass) & \frac{p-\Wass}{1-2\Wass}\Wass \\
\frac{1-p-\Wass}{1-2\Wass}\Wass & \frac{1-p-\Wass}{1-2\Wass}(1-\Wass)
\end{array}\right].\label{eq_35}
\end{align}
For this final distribution, the protocol that achieves the equality in (\ref{eq_8}) is realized by transporting probabilities along the optimal transport plan from \( p_0^{XY} \) to \( p_\tau^{XY} \) under a constant and uniform thermodynamic force \( F \). Given the protocol \(\{R_t^{X|y}(x,x')\}\) (the detailed form of the protocol is provided in the Appendix), we can calculate entropy production \(\Sigma_\tau^X\) for it. A numerical calculation of \(\Sigma_\tau^X\) for each \(\abs{\Delta I_\tau^{X:Y}}\) is shown as the orange line in Fig.~\ref{fig_num1} (a).
\begin{figure*}[tbp]
\centering
\includegraphics[scale=0.8]{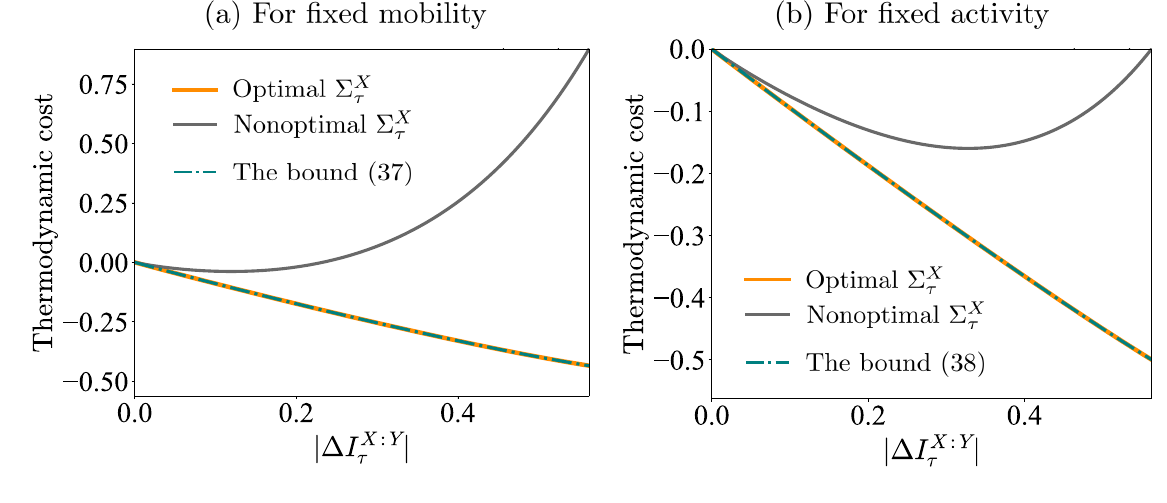}
\caption{(a) Numerical results with fixed mobility. entropy production for non-optimal final distributions does not achieve the bound (\ref{eq_32}), whereas entropy production for optimal final distributions achieves it. The parameters used are \(p = 0.75\), \(\langle m \rangle_\tau = 1\), and \(\tau = 1\). (b) Numerical results with fixed activity. entropy production for non-optimal final distributions does not achieve the bound (\ref{eq_33}), whereas entropy production for optimal final distributions does achieve the bound (\ref{eq_33}). The parameters used are \(p = 0.75\), \(\langle a \rangle_\tau = 1\), and \(\tau = 1\).}
\label{fig_num1}
\end{figure*}
Comparing this result with the right-hand side of (\ref{eq_32}) computed for each \(\abs{\Delta I_\tau^{X:Y}}\) (shown as the green dotted line in Fig.~\ref{fig_num1} (a)), it is evident that this protocol achieves the equality in the bound (\ref{eq_33}). Notably, the fact that entropy production is negative indicates that the feedback using information enables achieving an entropy production smaller than that imposed by the conventional thermodynamic second law, \(\Sigma_\tau^X \geq 0\). In other words, it implies the ability to extract work exceeding the negative of the change in nonequilibrium free energy.  

Additionally, we also show examples of non-optimal protocols. Specifically, we consider the case where the equality in (\ref{eq_8}) is achieved but the equality in (\ref{eq_16}) is not. For instance, for the same initial distribution, constructing a final distribution of the form 
\begin{align}
p_\tau^{XY}=\left[
\begin{array}{cc}
p-\Delta & 0 \\
\Delta & 1-p
\end{array}\right],\label{eq_36}
\end{align}
by choosing \(\Delta\) such that the mutual information change remains fixed at \(\abs{\Delta I_\tau^{X:Y}}\) yields a non-optimal final distribution for consuming \(\abs{\Delta I_\tau^{X:Y}}\). For this distribution, numerical calculations of entropy production \(\Sigma_\tau^X\) under a protocol \(\{R_t^{X|y}(x,x')\}\) that achieves the equality in (\ref{eq_8}) (the detailed form of the protocol is provided in the Appendix) are shown as the gray line in Fig.~\ref{fig_num1}(a). Since it does not achieve the equality in (\ref{eq_32}), it is evident that this is not an optimal feedback protocol for consuming \(\abs{\Delta I_\tau^{X:Y}}\).
                    
We next numerically demonstrate the bound (\ref{eq_33}) for fixed activity. For the initial distribution given by (\ref{eq_34}), we can take an optimal final distribution given by (\ref{eq_35}) and a non-optimal final distribution given by (\ref{eq_36}). For these final distributions, the protocols that achieve the equality in (\ref{eq_8}) transport probabilities from \(p_0^{XY}\) to \(p_\tau^{XY}\) along the optimal transport plan under a uniform and constant thermodynamic force. When \(D\) represents the time-averaged activity $\langle a\rangle_\tau$ fixed as $\langle a\rangle_\tau=A$, maintaining a constant thermodynamic force corresponds to setting the probability current \(J_t^{X|y}(x,'x)\) to a constant value $J$. Numerical calculations of entropy production \(\Sigma_\tau^X\) for such a protocol, \(\{R_t^{X|y}(x,x')\}\) (the detailed form of the protocol is provided in the Appendix), are shown in Fig.~\ref{fig_num1} (b). The orange line corresponds to the optimal final distribution, while the gray line corresponds to the non-optimal final distribution. Comparing these results with the right-hand side of (\ref{eq_33}) computed for each \(\abs{\Delta I_\tau^{X:Y}}\) (shown as the green dotted line in Fig.~\ref{fig_num1}(a)), it can be seen that the bound (\ref{eq_33}) is achieved for the optimal final distribution, while it is not achieved for the non-optimal final distribution.
\section{Conclusion}\label{sec_6}
In this study, we have derived the main Theorem (\ref{eq_16}) which provides the upper bound on the change in mutual information under the fixed error-free initial distribution, the fixed marginal distribution of the controller, and the fixed Wasserstein distance. The equality can be achieved under the condition (C) by setting the optimal final distribution. This Theorem gives the maximum change in mutual information that can be induced under a fixed minimum dissipation when the controller performs feedback control based on error-free measurement outcomes. Based on this Theorem, we have derived the speed limits (\ref{eq_32}) and (\ref{eq_33}) which give the lower bounds on entropy production for consuming a fixed amount of information through the feedback process. These bounds are achievable under the condition (C), and we have identified the optimal protocols. While a similar optimization problem for measurement processes has been solved in our previous work \cite{nagase2024thermodynamically}, the present work for feedback processes is based on largely different mathematical techniques including the achievable Fano's inequality (\ref{eq_9})~\cite{Sakai_e22030288_2020}, and therefore separate treatments are required for measurement and feedback processes.\par
Our results theoretically advance thermodynamics of information by specifying the form of contribution of processed information to finite-time entropy production. The optimal feedback protocols identified in this study can be applied to a wide range of feedback processes described by Markov jump processes, which would be experimentally implemented by single electron devices~\cite{koski2014experimental,koski2015chip,chida2017power}, double quantum dots in the classical regime~\cite{Schaller_PhysRevB.82.041303_2010,annby2020maxwell}, and effective discrete dynamics by Brownian nanoparticles~\cite{gingrich2017inferring}. Future directions also include extending our theoretical framework to continuous or quantum systems. 
\section*{ACKNOWLEDGEMENTS}
We thank Kosuke Kumasaki and Yosuke Mitsuhashi for valuable discussions. R.N. is supported by the World-leading Innovative Graduate Study Program for Materials Research, Industry, and Technology (MERIT-WINGS) of the University of Tokyo. T.S. is supported by Japan Society for the Promotion of Science (JSPS) KAKENHI Grant No. JP19H05796, JST, CREST Grant No. JPMJCR20C1 and JST ERATO-FS Grant No. JPMJER2204. T.S. is also supported by the Institute of AI and Beyond of the University of Tokyo and JST ERATO Grant No. JPMJER2302, Japan. 

%
%
%
%
%
\begin{appendix}
  \section{Derivation of the main Theorem and the speed limits}
  We provide the full proof of the main Theorem (\ref{eq_16}) and the speed limits (\ref{eq_32}) and (\ref{eq_33}).
  \subsection{Majorization}
We introduce the concept of majorization, which is used to provide an upper bound on mutual information. For two probability distributions \( p^X \) and \( q^X \), we denote that \( p^X \) majorizes \( q^X \), or \( q^X \prec p^X \), when they satisfy~\cite{bhatia2013matrix,Sagawa_Majorization_2022}
\begin{equation}
q^X\prec p^X\Longleftrightarrow\forall x\in\mathcal{X},\ \sum_{x'=1}^{x}q_\downarrow^X(x')\leq \sum_{x'=1}^{x}p_\downarrow^X(x').
\end{equation}
Here, \( p_\downarrow^X \) and \( q_\downarrow^X \) are the distributions obtained by rearranging \( p^X \) and \( q^X \) in descending order, respectively. When \( q^X \prec p^X \), the following relationship holds for any convex function \( f \)~\cite{bhatia2013matrix,Sagawa_Majorization_2022}:  
\begin{equation}
\sum_{x\in\mathcal{X}}f\left(q^X(x)\right)\leq\sum_{x\in\mathcal{X}}f\left(p^X(x)\right).
\end{equation}
By setting $f(x)=x\ln x$ and taking the negative of the both sides, we obtain the following relation of the Shannon enropy:
\begin{equation}
q^X\prec p^X\implies S(q^X)\geq S(p^X).\label{eq_S3}
\end{equation}
  \subsection{Proof of the main Theorem}
The proof of the main Theorem is structured into two propositions.
\begin{prop}\label{prop_A1}
$\Wass\left(p_0^{XY},p_\tau^{XY}\right)\geq \mathcal{E}.$
\end{prop}
\begin{proof}
Since the Wasserstein distance is greater than or equal to the total variation distance, we obtain
\begin{equation}
\Wass\left(p_0^{XY},p_\tau^{XY}\right)\geq\frac{1}{2}\sum_{x,y}\abs{p^{XY}_\tau(x,y)-p^{XY}_0(x,y)}.
\end{equation}
Since the marginal distribution of \( Y \) is fixed as \( p_0^Y(y) = p_\tau^Y(y) = p^Y(y) \), the absolute value on the right-hand side can be calculated as
\begin{equation}
\abs{p^{XY}_\tau(x,y)-p^{XY}_0(x,y)}=
\begin{cases}
p^{Y}(y)-p^{XY}_\tau(y,y) & x=y, \\
p^{XY}_\tau(x,y) & x\neq y.
\end{cases}
\end{equation}
By using $p^{Y}(y)-p^{XY}_\tau(y,y)=\sum_{x(\neq y)}p^{XY}_\tau(x,y)$, we get
\begin{align}
\frac{1}{2}\sum_{x,y}\abs{p^{XY}_\tau(x,y)-p^{XY}_0(x,y)}=\sum_{x\neq y}p^{XY}_\tau(x,y)=\mathcal{E}.
\end{align}
\end{proof}
\begin{prop}\label{prop_A2}
$S\left(\tilde{p}^Y_\mathcal{E}\right)$ is monotonically and strictly increasing in $\mathcal{E}\in[0,1-p_\downarrow^Y(1)]$.
\end{prop}
\begin{proof}
Recalling the definition of \(M_\Error\), if there exists a maximum \(m \in \{1, 2, \cdots, n\}\) that satisfies \(\sum_{y=1}^{m} p_\downarrow^Y(y) - (m-1)p_\downarrow^Y(m) < 1 - \mathcal{E}\), this \(m\) is defined as \(M_\Error\). First, we prove that \(M_\Error\) decreases with respect to \(\mathcal{E}\). Let \(f(m) \coloneqq \sum_{y=1}^{m} p_\downarrow^Y(y) - m p_\downarrow^Y(m)\). For \(m' > m\), we have
  \begin{align}
    f(m')&=\sum_{y=1}^{m'}p_\downarrow^Y(y)-m'p_\downarrow^Y(m')\\
    &=\sum_{y=1}^{m'}\left[p_\downarrow^Y(y)-p_\downarrow^Y(m')\right]\\
    &\geq\sum_{y=1}^{m'}\left[p_\downarrow^Y(y)-p_\downarrow^Y(m)\right]\\
    &\geq\sum_{y=1}^{m}\left[p_\downarrow^Y(y)-p_\downarrow^Y(m)\right]\\
    &=f(m).
  \end{align}
Therefore, \( f(m) \) increases with respect to \( m \). Consequently, \( M_\Error \), which is determined as the maximum \( m \) satisfying \( f(m) < 1 - \mathcal{E} \), decreases with respect to \(\mathcal{E}\).\par
We next recall the definition of $\tilde{p}_{\Error}^Y$:
  \begin{align}
    \tilde{p}_{\Error}^Y=
    \begin{cases}
      p^Y(y), & \sigma_{p^Y}(y)> M_\Error, \\
      \displaystyle\frac{\sum_{y=1}^{M_\Error}p_\downarrow^Y(y')-(1-\Error)}{M_\Error-1} & 2\leq\sigma_{p^Y}(y)\leq M_\Error\\
      1-\Error, & \sigma_{p^Y}(y)=1.
    \end{cases}
  \end{align}
We note that the permutation \(\sigma_{p^Y}\), which rearranges \(p^Y\) in descending order, coincides with \(\sigma_{\tilde{p}_{\mathcal{E}}^Y}\), the permutation rearranging \(\tilde{p}_{\mathcal{E}}^Y\) in descending order. Now, consider \(\mathcal{E}' > \mathcal{E}\). We aim to show that \(\tilde{p}^Y_{\mathcal{E}'} \prec \tilde{p}^Y_{\mathcal{E}}\), i.e.,  
  \begin{align}\label{eq_B_9}
    \forall y,\quad \sum_{y'=1}^y\tilde{p}_{\Error,\downarrow}^Y(y')\geq\sum_{y'=1}^y\tilde{p}_{\Error',\downarrow}^Y(y').
  \end{align}
First, when \( y \geq M_\mathcal{E} \), \( y \geq M_{\mathcal{E}'} \) also holds. Therefore, we have
  \begin{align}
    \sum_{y'=1}^y\tilde{p}_{\Error,\downarrow}^Y(y')=\sum_{y'=1}^y\tilde{p}_{\Error',\downarrow}^Y(y')=\sum_{y'=1}^yp_{\downarrow}^Y(y').
  \end{align}
When $M_{\Error'}\leq y\leq M_\Error$,
  \begin{align}
    \sum_{y'=1}^y\tilde{p}_{\Error,\downarrow}^Y(y')&=\sum_{y'=1}^{M_\Error}p_{\downarrow}^Y(y')-(M_\Error-y)\tilde{p}_{\Error,\downarrow}^Y(M_\Error),\\
    \sum_{y'=1}^y\tilde{p}_{\Error',\downarrow}^Y(y')&=\sum_{y'=1}^{y}p_{\downarrow}^Y(y')
  \end{align}
holds, which yields
  \begin{align}
    &\sum_{y'=1}^y\tilde{p}_{\Error,\downarrow}^Y(y')-\sum_{y'=1}^y\tilde{p}_{\Error',\downarrow}^Y(y')\\
    &=\sum_{y'=y+1}^{M_\Error}p_{\downarrow}^Y(y')-(M_\Error-y)\tilde{p}_{\Error,\downarrow}^Y(M_\Error)\\
    &=\sum_{y'=y}^{M_\Error}\left[p_{\downarrow}^Y(y')-\tilde{p}_{\Error,\downarrow}^Y(M_\Error)\right]\\
  &\geq0.
  \end{align}
Additionally, when \( 2 \leq y \leq M_{\mathcal{E}'} \),
  \begin{align}
    \sum_{y'=1}^y\tilde{p}_{\Error,\downarrow}^Y(y')&=\sum_{y'=1}^{M_\Error}p_{\downarrow}^Y(y')-(M_\Error-y)\tilde{p}_{\Error,\downarrow}^Y(M_\Error),\\
    \sum_{y'=1}^y\tilde{p}_{\Error',\downarrow}^Y(y')&=\sum_{y'=1}^{M_{\Error'}}p_{\downarrow}^Y(y')-(M_{\Error'}-y)\tilde{p}_{\Error',\downarrow}^Y(M_{\Error'})\\
  \end{align}
holds. Therefore,
  \begin{align}
   &\sum_{y'=1}^y\tilde{p}_{\Error,\downarrow}^Y(y')-\sum_{y'=1}^y\tilde{p}_{\Error',\downarrow}^Y(y)
   \\&=\sum_{y'=M_{\Error'}+1}^{M_\Error}\left[p_{\downarrow}^Y(y')-\tilde{p}_{\Error,\downarrow}^Y(M_\Error)\right]\notag\\
    &\hspace{50pt}+\sum_{y'=2}^{M_{\Error'}}\left[\tilde{p}_{\Error',\downarrow}^Y(M_{\Error'})-\tilde{p}_{\Error,\downarrow}^Y(M_\Error)\right]\\
  &\geq0.
  \end{align}
From the above, it follows that \(\tilde{p}^Y_{\mathcal{E}'} \prec \tilde{p}^Y_{\mathcal{E}}\). Moreover, by definition, when \(\mathcal{E}' \neq \mathcal{E}\), \(\tilde{p}^Y_{\mathcal{E}'} \neq \tilde{p}^Y_{\mathcal{E}}\) holds.
\end{proof}
From Propositions \ref{prop_A1} and \ref{prop_A2}, the upper bound (\ref{eq_16}) on the consumed mutual information is obtained when \(\Wass\left(p_0^{XY}, p_\tau^{XY}\right) = \Wass\) is fixed.
  \subsection{Formalism by Lorenz curves}
  While \(\tilde{p}_\Error^Y\) was constructed using histograms of probability distributions in Sec.\ref{sec_3}, this construction can also be explained using Lorenz curves. A Lorenz curve is a polyline uniquely associated with a probability distribution and is constructed as follows~\cite{Sagawa_Majorization_2022}. First, arrange the probability distribution \(p^Y\) in descending order and obtain \(p_\downarrow^Y\). Next, plot the cumulative probability \(\sum_{y'=1}^y p_\downarrow^Y(y')\) as a function of state \(y\). The polyline obtained by connecting these points in addition to the origin \((0, 0)\) is the Lorenz curve for the probability distribution \(p^Y\).  

The Lorenz curve visualizes the concentration of probabilities in the distribution. This can be understood as follows: consider two probability distributions \(p^Y\) and \(q^Y\), and assume that the Lorenz curve of \(p^Y\) lies above that of \(q^Y\). That is, for all \(y\), \(\sum_{y'=1}^y p_\downarrow^Y(y') \geq \sum_{y'=1}^y q_\downarrow^Y(y')\). Then, from the definition of majorization introduced in Appendix A1, \(p^Y \succ q^Y\) holds, which implies \(S(p^Y) \leq S(q^Y)\). Therefore, a probability distribution with a Lorenz curve higher on the graph has smaller Shannon entropy and is more concentrated.  

Using this property of Lorenz curves, the construction of \(\tilde{p}_\Error\) can be intuitively explained. First, if \(1 - \Error < p_\downarrow^Y(1)\), \(\tilde{p}_\Error = p^Y\). For \(1 - \Error \geq p_\downarrow^Y(1)\), first draw the Lorenz curve of \(p^Y\) (the black line in Fig.~\ref{fig_Lorenz} (a)). Next, extend a half-line upward and to the right from point \(\mathrm{P}_1\) defined by coordinate \((1, 1 - \Error)\) (the red point in Fig.~\ref{fig_Lorenz} (a)), and rotate it around \(\mathrm{P}_1\) to approach the Lorenz curve of \(p^Y\) from above (the red dotted line in Fig.~\ref{fig_Lorenz} (a)). When the half-line touches the Lorenz curve of \(p^Y\), stop the rotation, and call the point of contact \(\mathrm{P}_2\) (Fig.~\ref{fig_Lorenz} (b)). The state \(y\) corresponding to point \(\mathrm{P}_2\) is \(M_\Error\) as defined in Sec.~\ref{sec_3}. Finally, connect the three points (origin, \(\mathrm{P}_1\), and \(\mathrm{P}_2\)) with polyline, and extend the curve beyond \(\mathrm{P}_2\) along the Lorenz curve of \(p^Y\). The probability distribution corresponding to this modified Lorenz curve (the red line in Fig.~\ref{fig_Lorenz} (c)) is defined as \(\tilde{p}_\Error\).
  \begin{figure*}[tbp]
\centering
\includegraphics[scale=1]{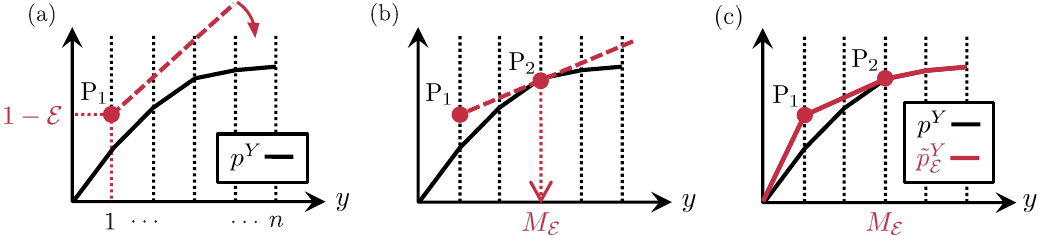}
\caption{(a) The Lorenz curve of \(p^Y\) and the half-line extended from point \(\mathrm{P}_1\) with coordinate \((1, 1 - \Error)\).  
(b) The construction of point \(\mathrm{P}_2\), where the half-line from \(\mathrm{P}_1\) touches the Lorenz curve of \(p^Y\).  
(c) The construction of the Lorenz curve for \(\tilde{p}_\Error\).  }
\label{fig_Lorenz}
\end{figure*}
\section{Optimal protocol for a two level system}
In the example of the two-level system discussed in Sec.~\ref{sec_5}, the optimal transport from the initial distribution defined by Eq.~(\ref{eq_34}) to the optimal final distribution defined by Eq.~(\ref{eq_35}), as well as to the non-optimal final distribution defined by Eq.~(\ref{eq_36}), was constructed. To provide a general method for constructing these protocols, we present the protocol for optimal transport from the initial distribution defined by Eq.~(\ref{eq_34}) to the final distribution defined by
\begin{align}
  p_\tau^{XY}=\left[
  \begin{array}{cc}
    p-\Delta_1 & \Delta_2 \\
    \Delta_1 & 1-p-\Delta_2
  \end{array}\right].
\end{align}
\subsection{For fixed mobility}
The protocol that achieves the fundamental limit (\ref{eq_32}) for fixed mobility transports probabilities along the optimal transport plan from the initial distribution to the final distribution under a constant and uniform thermodynamic force \(F\). The optimal transport plan is to transport a probability \(\Delta_1\) from state \((1,1)\) to \((2,1)\) and a probability \(\Delta_2\) from state \((2,2)\) to \((1,2)\). 

To implement this using a constant thermodynamic force \(F\) by controlling the transition rates in a Markov jump process, the time interval \([0, \tau]\) is divided into two intervals: \(t \in [0, \tau/2]\) and \(t \in [\tau/2, \tau]\). In individual intervals, the transports of \(\Delta_1\) and \(\Delta_2\) are performed.  

First, in the interval \(t \in [0, \tau/2]\), only the transport of probability \(\Delta_1\) from state \((1,1)\) to \((2,1)\) is carried out. An example of the time evolution implemented in this interval is  
\begin{align}
  p_t^{XY}(2,1)&=\frac{\Delta_1t}{\tau/2},\\
  p_t^{XY}(1,1)&=p-\frac{\Delta_1t}{\tau/2},
\end{align}
while keeping the probabilities $p_t^{XY}(2,1)$ and $pt^{XY}(2,2)$. To implement this under a constant thermodynamic force \( F \), the transition rates \( R_t^{X|1}(x, x') \) should be set as
\begin{align}
  R_t^{X|1}(2,1)&=\frac{\Delta_1}{\tau/2}\frac{1}{1-e^{-F}}\frac{1}{p-\frac{\Delta_1t}{\tau/2}},\\
  R_t^{X|1}(1,2)&=\frac{e^{-F}}{1-e^{-F}}\frac{1}{t}.
\end{align}
The remaining transition rates are set to \( R_t^{X|2}(x, x') = 0 \).\par
In the interval \(t \in [\tau/2, \tau]\), the probability transport of \(\Delta_2\) from state \((2,2)\) to \((1,2)\) is performed. The example of the time evolution to be implemented in this case is
\begin{align}
  p_t^{XY}(1,2)&=\frac{\Delta_2t}{\tau/2},\\ 
  p_t^{XY}(1,1)&=1-p-\frac{\Delta_2t}{\tau/2}
\end{align}
with \(p_t^{XY}(1,1)\) and \(p_t^{XY}(2,1)\) kept constant. To implement this under a constant and uniform  thermodynamic force \(F\), the transition rates \(R_t^{X|2}(x,x')\) should be set as
\begin{align}
  R_t^{X|2}(1,2)&=\frac{\Delta_2}{\tau/2}\frac{1}{1-e^{-F}}\frac{1}{1-p-\frac{\Delta_2t}{\tau/2}},\\
  R_t^{X|2}(2,1)&=\frac{e^{-F}}{1-e^{-F}}\frac{1}{t}.
\end{align}
The remaining transition rates are set to \(R_t^{X|1}(x,x') = 0\). Moreover, the entropy production generated by the entire protocol over the two intervals is given by
\begin{align}
  \Sigma_\tau^{XY}&=\int_0^\tau\diff t\sum_{x\neq x',y}J_t^{X|y}(x,x')F_t^{X|y}(x,x')\label{eq_B12}\\
  &=(\Delta_1+\Delta_2)F\\
  &=\frac{\Wass(p_0^{XY},p_\tau^{XY})^2}{\langle m\rangle_\tau},
\end{align}
which achieves the bound (\ref{eq_32}). The orange solid line and gray solid line representing entropy production \(\Sigma_\tau^X\) in Fig.~\ref{fig_num1} (a) are plotted using \(\Sigma_\tau^{XY}\) calculated from Eq.~(\ref{eq_B12}), with \(\Sigma_\tau^X = \Sigma_\tau^{XY} - \abs{\Delta I_\tau^{X:Y}}\).
\subsection{For fixed activity}
Next, we present the protocol that achieves the fundamental bound (\ref{eq_33}) for fixed activity. This protocol transports probabilities along the optimal transport plan from the initial distribution to the final distribution under a constant and uniform activity \(A\) and probability current \(J\). The optimal transport plan, as well as the case of fixed mobility, transports a probability \(\Delta_1\) from state \((1,1)\) to \((2,1)\) and a probability \(\Delta_2\) from state \((2,2)\) to \((1,2)\).  

To achieve this under a constant activity \(A\) and probability current \(J = (\Delta_1 + \Delta_2)/\tau\) (which is uniquely determined by the fixed time \(\tau\) and the initial and final distributions), the time interval \([0, \tau]\) is divided into two intervals 
\(t \in \left[0, \frac{\Delta_1}{\Delta_1 + \Delta_2}\tau\right]\) and \(t \in \left[\frac{\Delta_1}{\Delta_1 + \Delta_2}\tau, \tau\right]\). In individual intervals, the transports of \(\Delta_1\) and \(\Delta_2\) are performed.\par
First, in the interval \(t \in \left[0, \frac{\Delta_1}{\Delta_1 + \Delta_2}\tau\right]\), the transport of a probability \(\Delta_1\) from state \((1,1)\) to \((2,1)\) is performed. The time evolution to be implemented is  
\begin{align}
  p_t^{XY}(2,1)=\frac{(\Delta_1+\Delta_2)t}{\tau},\quad p_t^{XY}(1,1)=p-\frac{(\Delta_1+\Delta_2)t}{\tau}
\end{align}
with \(p_t^{XY}(2,1)\) and \(p_t^{XY}(2,2)\) kept constant. To implement this under a constant activity \(A\) and probability current \(J = (\Delta_1 + \Delta_2)/\tau\), the transition rates \(R_t^{X|1}(x,x')\) should be set as 
\begin{align}
  R_t^{X|1}(2,1)&=\frac{A+J}{2(p-Jt)},\\
  R_t^{X|1}(1,2)&=\frac{A-J}{2Jt}.
\end{align}
The remaining transition rates are set to \(R_t^{X|2}(x,x') = 0\).\par
In the interval \(t \in \left[\frac{\Delta_1}{\Delta_1 + \Delta_2}\tau, \tau\right]\), the transport of a probability \(\Delta_2\) from state \((2,2)\) to \((1,2)\) is conducted. The time evolution to be implemented in this interval is
\begin{align}
  p_t^{XY}(1,2)&=\frac{(\Delta_1+\Delta_2)t}{\tau},\\
  p_t^{XY}(1,1)&=1-p-\frac{(\Delta_1+\Delta_2)t}{\tau},
\end{align}
with \(p_t^{XY}(1,1)\) and \(p_t^{XY}(2,1)\) kept constant. To implement this under a constant activity \(A\) and probability current \(J = (\Delta_1 + \Delta_2)/\tau\), the transition rates \(R_t^{X|2}(x,x')\) should be set as  
\begin{align}
  R_t^{X|2}(1,2)&=\frac{A+J}{2(p-Jt)},\\
  R_t^{X|2}(2,1)&=\frac{A-J}{2Jt}.
\end{align}
The remaining transition rates are set to \(R_t^{X|1}(x,x') = 0\). The total entropy production generated by the entire protocol is calculated as
\begin{align}
  \Sigma_\tau^{XY}&=\int_0^\tau\diff t\sum_{x\neq x',y}J_t^{X|y}(x,x')F_t^{X|y}(x,x')\label{eq_B21}\\
  &=2(\Delta_1+\Delta_2)\tanh^{-1}\left(\frac{\Delta_1+\Delta_2}{A\tau}\right)\\
  &=\Wass(p_0^{XY},p_\tau^{XY})\tanh^{-1}\frac{\Wass(p_0^{XY},p_\tau^{XY})}{\langle a\rangle_\tau},
\end{align}
which achieves the bound (\ref{eq_33}). The orange solid line and gray solid line representing entropy production \(\Sigma_\tau^X\) in Fig.~\ref{fig_num1} (b) are plotted using \(\Sigma_\tau^{XY}\) calculated from Eq.~(\ref{eq_B21}), with \(\Sigma_\tau^X = \Sigma_\tau^{XY} - \abs{\Delta I_\tau^{X:Y}}\).
\end{appendix}
%

\end{document}